\documentclass[prx,prl,a4paper,aps,twocolumn,superscriptaddress,longbibliography]{revtex4-1}
\usepackage{amssymb,amsfonts,bm}
\usepackage{graphicx}
\usepackage{epstopdf}
\usepackage{times}
\usepackage{float}
\usepackage{lipsum}
\usepackage{color}
\usepackage{dcolumn}
\usepackage{bm}
\usepackage{subfigure}
\usepackage[colorlinks,linkcolor=blue,anchorcolor=blue,urlcolor=blue,citecolor=blue]{hyperref}
\usepackage{verbatim} 
\usepackage{amsmath} 
\usepackage{tikz}
\usepackage{braket}
\usepackage{bbold}
\usepackage{tikz} 
\usetikzlibrary{quantikz}

\usepackage{amsthm}
\newtheorem{theorem}{Theorem}
\newtheorem{definition}[theorem]{Definition} 

\newtheorem{corollary}[theorem]{Corollary}
\newtheorem{example}[theorem]{Example}
\newtheorem{proposition}[theorem]{Proposition}
\newtheorem{remark}{Remark}

\usepackage{hyperref}
\hypersetup{colorlinks=true, linkcolor=blue, citecolor=red, urlcolor=blue  }

\usepackage{physics}

\usepackage{algorithm}
\usepackage{algpseudocode}

\allowdisplaybreaks[4]
\usepackage[normalem]{ulem}

\begin{document}

\title{A Lie-algebraic Criterion for the Universality of Exponentiated Quantum Gates}

\author{Yinuo Xue}
\email{yinuoxue@link.cuhk.edu.cn}
\affiliation{Mathematics Research Center, School of Science and Engineering, The Chinese University of Hong Kong, Shenzhen, China}

\author{Qian Chen}
\email{chenqian.phys@gmail.com}
\affiliation{Mathematics Research Center, School of Science and Engineering, The Chinese University of Hong Kong, Shenzhen, China}

\author{Jing-Song Huang}
\email{corresponding author: huangjingsong@cuhk.edu.cn}
\affiliation{Mathematics Research Center, School of Science and Engineering, The Chinese University of Hong Kong, Shenzhen, China}


\begin{abstract}

We present a criterion that serves as the basis for a polynomial-time algorithm to decide whether a finite set of qudit gates exponentiated by some Hamiltonians is universal. Our approach formulates universality in Lie algebraic terms and applies Borel--de Siebenthal theory with a diagonal generator having incommensurate spectrum. In this framework, nonuniversality is detected by invariant subspaces, equivalently by a graph-connectivity obstruction, while universality is repaired by adding generators that couple disconnected components. We further prove that two generators are sufficient for universal control. Our work reveals a profound link between qudit universality and irreducibility of Lie algebra representations.

\end{abstract}

\maketitle



Universality is the foundational premise of quantum computation, guaranteeing that any target unitary operation can be approximated to arbitrary precision by a finite sequence of available operations.
While the conditions for universal control have been comprehensively characterized in the qubit regime, higher-dimensional quantum systems---qudits---are increasingly recognized for their potential to transcend qubit limitations. Qudit architectures promise substantially increased information density, simplified error correction, and the efficient native realization of complex algorithmic primitives \cite{oszmaniec2016random,sawicki2014convexity}. Consequently, developing rigorous criteria to verify the universality of a given $d$-dimensional ($d \geq 2$) control set has emerged as a critical challenge in quantum information science.

Historically, universality criteria was framed in a Lie group-theoretic term, based on finite sets of unitary gates. However, this approach introduces two practical bottlenecks. First, verifying universality from discrete matrix generators typically incurs a prohibitive computational overhead, frequently scaling beyond polynomial time \cite{brennen2005criteria,sawicki2017universality,van2021universality}, rendering these criteria intractable for high-dimensional verification tasks. Second, finite gate analysis is fundamentally decoupled from the physical reality of near-term and fault-tolerant quantum hardware. In practice, unitary gates are not directly implemented as discrete primitives; instead, they arise from time evolution under hardware-native Hamiltonians \cite{albertini2003notions}.

In this Letter, we resolve the above problems by developing a systematic, polynomial-time algorithm for verifying single-qudit universality Lie algebraic framework. We focus on single-qudit gates, since universality for multi-qudit gates can be achieved by a single-qudit universal set plus an entangling two-qudit gate \cite{Brylinski2002MQC,Oszmaniec:2017cyf}. Rather than asking whether a discrete sequence of abstract gates densely covers $U(d)$ or $SU(d)$, we pose a fundamentally operational question: \textit{under what conditions does the continuous-time evolution generated by exponentiating a set of available skew-Hermitian generators natively yield a universal gate family?}

Our approach leverages a ubiquitous feature of physical quantum systems: the presence of at least one diagonal Hamiltonian with incommensurate spectrum, serving as a ``general direction'' in the Lie algebra. We show that the continuous evolution of this single element ensures that the generated group is dense in a maximal torus of the target group. By anchoring our system to this maximal torus, we unlock the application of the Borel--de Siebenthal classification theory for compact Lie groups \cite{BorelDeSiebenthal1949}. This theorem, when combined with the existence of a maximal torus, proves that a system fails to be universal if and only if its generators exhibit a nontrivial block-diagonal form.

This leads us to a simple operational criterion of this Letter. Consequently, universality of a connected Lie group is detected by absence of invariant coordinate subspaces under the action of the remaining off-diagonal Hamiltonians. This immediately translates the universality problem into a highly efficient graph-connectivity algorithm that runs in polynomial time with respect to both the system dimension $d$ and the number of generators. Furthermore, our Lie algebraic approach is inherently constructive: if a hardware topology is found to be non-universal (reducible), our framework explicitly computes the missing generators required to repair the set and achieve full universality. As a fundamental consequence of this graph-theoretic reduction, we are able to rigorously establish that a minimal topology consisting of exactly two Hamiltonians --- a diagonal drift and a single connected control---is sufficient to achieve full universality.





\section{The Lie algebraic criterion}
\label{sec:LieTheory}

This section provides the necessary background for our main results. We begin with the definition of universality for a Lie group and then discuss its connection to the corresponding notion for its Lie algebra. We consider the scenario where the set of generators includes a diagonal Hamiltonian with incommensurate spectrum — i.e., a general element that generates a maximal torus. According to the Borel--de Siebenthal classification, any connected closed subgroup of $U(d)$ or $SU(d)$ that contains a maximal torus must either be the full group or preserve a nontrivial block-diagonal structure. This translates the verification of universality into checking whether the defining representation (on $\mathbb{C}^d$) of the generated Lie algebra is irreducible.

\subsection{The universality of a Lie group and its Lie algebra}

We define universality of a Lie group $G$ and that of its Lie algebra $\mathfrak{g} \equiv \operatorname{Lie}(G)$.
\begin{definition}[Universality of a Lie group]
A finite set of gates $\mathcal{S}=\{U_1,\dots,U_m\} \subset G$ where $G$ is a Lie group, is said to be universal if the following set generated by elements of $\mathcal{S}$ via group multiplication
\begin{equation}
  \langle \mathcal{S} \rangle = \left\{U_{i_{1}}U_{i_{2}} \cdots U_{i_{m}} \vert U_{i_{j}} \in \mathcal{S}, m\in \mathbb{N}  \right\}  
\end{equation}
is dense in $G$, i.e., the closure $\overline{\langle \mathcal{S} \rangle} =G$. In fact $\overline{\langle \mathcal{S} \rangle}$ is a Lie group \cite{sawicki2017universality}.
\end{definition}

In other words, every element of $G$ can be approximated arbitrarily well by finite products of elements of $\mathcal S$.

\begin{definition}
Consider a finite set of elements in Lie algebra $\mathfrak{g}$, denoted as $\mathcal{X}=\{X_1,\dots,X_m\} \subset \mathfrak{g}$, where $\mathfrak{g}$ is either $\mathfrak{u}(d)$ or $\mathfrak{su}(d)$.
Consider the following set generated by elements of $\mathcal{X}$ via Lie bracket commutators
\begin{equation}
\mathfrak{l}:=\operatorname{Lie}_{\mathbb{R}}\langle \mathcal{X} \rangle= \operatorname{Lie}_{\mathbb{R}} \langle X_1,\dots,X_m\rangle
\subset \text{End}_{\mathbb R}(V).
\label{eq:GeneratingLieAlgebra}
\end{equation}
Note that $\mathfrak{l}$ is the smallest real Lie algebra containing $\mathcal{X}$.    
\end{definition}

From now on, we will be interested in the cases $G=U(d)$ and $G=SU(d)$. It is natural to study universality at the Lie algebra level, since quantum gates are implemented by exponentiating Hamiltonians. Consider the small-step gate set
\[
S_\varepsilon=\{\exp(\varepsilon X_1),\dots,\exp(\varepsilon X_m)\}.
\]
For sufficiently small $\varepsilon>0$, the closure of $\langle S_\varepsilon \rangle$, denoted as $H_{\varepsilon}=\overline{\langle \mathcal{S}_{\varepsilon} \rangle}$, has Lie algebra $\operatorname{Lie}(H_{\varepsilon})$ containing $\operatorname{Lie}\langle \mathcal X\rangle$; hence, if $\operatorname{Lie} \langle \mathcal X\rangle=\mathfrak g$, then $\operatorname{Lie}(H_{\varepsilon})=\mathfrak{g}$, implying that $S_\varepsilon$ is universal (see Proposition~\ref{SM:prop:small-step} in Supplemental Material in details).

Following this observation, we construct the set $S_\varepsilon \subseteq G$ from its Lie algebra. The closure of $S_\varepsilon$ is a connected closed subgroup of $G$, and the Borel--de Siebenthal classification allows us to analyze its structure and thereby determine whether $S_\varepsilon$ is universal.
  



\subsection{Universality and irreducibility}
\label{subsec:Borel-deSienbentalTheory}

\paragraph*{Maximal tori and general elements/directions.}


Let $G$ be a compact connected Lie group. A maximal torus in $G$ is a closed, connected, abelian subgroup maximal under inclusion. For such a group, maximal tori exist and are unique up to conjugation \cite{BroeckerTomDieck1985}, and their common dimension is called the rank of $G$. A closed subgroup $H\subset G$ is said to have maximal rank if it contains a maximal torus of $G$.

For $G=U(d)$ we use the standard maximal tori
\[
T_{U(d)} =
\bigl\{ \operatorname{diag}(e^{i\theta_{1}},\dots,e^{i\theta_{d}}) :
\theta_{j}\in\mathbb{R} \bigr\},
\]
with $\operatorname{rank} U(d) = d$. Likewise for $G=SU(d)$,
\[
T_{SU(d)} =
\bigl\{ \operatorname{diag}(e^{i\theta_{1}},\dots,e^{i\theta_{d}}) :\theta_{1}+\cdots+\theta_{n} \equiv 0 \ (\mathrm{mod}\ 2\pi)\bigr\}
\]
with each $\theta_{j}\in\mathbb{R}$, and $\operatorname{rank} SU(d) = d-1$.
\


\begin{definition}[General elements and general directions]
\label{def:GeneralElements}
Let $G$ be a compact connected Lie group. An element $g\in G$ is called \emph{general} if the closure of the cyclic subgroup, generates is a maximal torus, i.e. $$\overline{\langle g\rangle}=T,
\quad \text{where} \quad
\langle g\rangle = \left\{g^{k} \mid k 
\in \mathbb{N}\right\},$$ 
for some maximal torus $T\subset G$.
An element $X\in\operatorname{Lie}(G)$ is called a \emph{general direction} if $\exp(X)$ is general.

\end{definition}

\noindent In the case of a standard torus $T^d = \mathbb{R}^d / 2\pi\mathbb{Z}^d$, general elements are characterized by linear independence over~$\mathbb{Q}$ of the angular components --- a form of the Kronecker-Weyl density theorem \cite{BroeckerTomDieck1985}. For the $U(d)$ and $SU(d)$ cases, this condition admits an explicit characterization.

\smallskip

Any diagonal element of $U(d)$ can be written as $\operatorname{diag}(e^{i\phi_1},\dots,e^{i\phi_d}) \in U(d)$ with some $\phi_{i} \in \mathbb{R}$. Then a general element $g$ with respect to $U(d)$ or $SU(d)$ is defined as the following:
\begin{enumerate}
    \item For $U(d)$, if $1,\frac{\phi_1}{2\pi},\dots,\frac{\phi_d}{2\pi}$ are linearly independent over $\mathbb Q$, then $\overline{\langle g\rangle}=T_{U(d)}.$ Then $g$ is said to be general, and $X=\operatorname{diag}(i\phi_1,\dots,i\phi_d)\in\mathfrak u(d)$ a diagonal general direction.

    \item For $SU(d)$, let $g=\operatorname{diag}(e^{i\phi_1},\dots,e^{i\phi_d})\in SU(d)$ with $\phi_1+\cdots+\phi_d\equiv 0 \pmod{2\pi}$. If $1,\frac{\phi_1}{2\pi},\dots,\frac{\phi_{d-1}}{2\pi}$ are linearly independent over $\mathbb Q$, then $\overline{\langle g\rangle}=T_{SU(d)}.$ Then $g$ is said to be general, and $X=\operatorname{diag}(i\phi_1,\dots,i\phi_d)\in\mathfrak{su}(d)$ a diagonal general direction.

\end{enumerate}


Such diagonal general directions can be chosen explicitly and are built into our construction.
The linear independence over $\mathbb{Q}$ would not be affected by some scaling factor.
Further, $\exp(\varepsilon X_1)\in S_\varepsilon$ is general, then
\[
\overline{\langle \exp(\varepsilon X_1)\rangle}=T_G\subset H_\varepsilon.
\]
In the small-step regime, $H_\varepsilon$ is therefore a connected closed Lie subgroup of maximal rank. This places the problem in the setting of the maximal-rank classification used in the next subsection.

\paragraph*{Universality via general direction.}

We now state the maximal-rank structural consequence that underlies our criterion. The Borel--de Siebenthal classification, together with the unitary-group formulation of Borevich--Krupetskii \cite{BorelDeSiebenthal1949,BorevichKrupetskii1981}, shows that every connected closed maximal-rank subgroup of $U(d)$ or $SU(d)$ is, up to conjugation in the ambient group, of block-diagonal type. Since our algorithm is formulated in the standard basis, we need a coordinate-level version of this statement. Accordingly, we restrict to subgroups containing the standard maximal torus, in which case the remaining ambiguity reduces to a permutation of the coordinate axes. This yields the following coordinate-adapted block-diagonal description.

We then address the next question: if a connected closed subgroup of $U(d)$ contains the maximal torus of diagonal matrices, what structure does it necessarily admit? The answer is that the subgroup must be block-diagonal up to conjugation. This assertion follows from combining the Borel--de Siebenthal classification with the unitary-group formulation of the Borevich--Krupetskii \cite{BorelDeSiebenthal1949,BorevichKrupetskii1981}.
Since we have already adopted the standard basis, we give the following coordinate-adapted statement.

\begin{theorem}[Block-diagonal structure]\label{thm:block}
Let $H \subset G$ be a connected closed subgroup of $G = U(d)$ [resp. $SU(d)$]. If $H$ contains the standard maximal torus $T_G$, then there exists a partition $d = d_1 + \dots + d_r$ and a permutation matrix $P$ (representing a reordering of the standard basis) such that the conjugated subgroup $\tilde{H} = P H P^{-1}$ is exactly the block-diagonal subgroup:
\begin{equation}
    \tilde{H} = \prod_{j=1}^r U(d_j) \quad \left[\text{resp. } S\left(\prod_{j=1}^r U(d_j)\right)\right].
\end{equation}
Furthermore, the Lie algebra of $\tilde{H}$, given by $\text{Lie}(\tilde{H}) = P \text{Lie}(H) P^{-1}$, inherits the same block-diagonal structure from $\tilde{H}$:
\begin{equation}
    \text{Lie}(\tilde{H}) = \bigoplus_{j=1}^{r}\mathfrak{u}(d_j) \quad \left[\text{resp. } \bigoplus_{j=1}^{r}\mathfrak{su}(d_j) \oplus i\mathbb{R}^{r-1}\right].  
\end{equation}
\end{theorem}

A rigorous proof, relying on root-space decompositions and the Borevich--Krupetskii formulation, is provided in Theorem~\ref{SM:thm:block} of the Supplemental Material. 

There is a direct observation that the required block-diagonal structure is presented under standard basis vectors up to a permutation. In other words, whenever a connected closed maximal-rank subgroup contains the standard maximal torus, one may permute the standard basis so that both the subgroup and its Lie algebra become block-diagonal with respect to the same coordinate decomposition.

However, it should be noted that Theorem~\ref{thm:block} applies only to connected closed subgroups, whereas the closed subgroup generated by a finite gate set need not be connected, since discrete permutations may arise from permutations between blocks of equal size. To circumvent such disconnectedness, we restrict the exponentiation parameter to the small-step regime, where the generated subgroup remains connected.

\begin{proposition}[Small-step connectedness]
\label{prop:smalleps}
Let $H \subset U(d)$ or $SU(d)$ be a closed maximal-rank subgroup with identity component $H^0$. If $H\neq H^0$, then every element of $H\setminus H^0$ lies at operator-norm distance at least $\sqrt{2}$ from the identity. Hence, for sufficiently small $\varepsilon>0$, the subgroup generated by
$S_\varepsilon=\{\exp(\varepsilon X_1),\dots,\exp(\varepsilon X_m)\}$
is connected, implying $H_\varepsilon = H_\varepsilon^0$.
\end{proposition}
\begin{proof}
See Proposition~\ref{SM:prop:smalleps} in Supplemental Material.
\end{proof}

Our criterion combines Theorem~\ref{thm:block} and Proposition~\ref{prop:smalleps}, leading to a condition linking universality to the absence of nontrivial invariant subspaces --- i.e., the irreducibility of $\operatorname{Lie}(H_\varepsilon)$ on $\mathbb{C}^d$. Consequently, universality of $S_\varepsilon$ is verified by examining the action of $\operatorname{Lie}(H_\varepsilon)$ on $\mathbb{C}^d$: non-universality occurs precisely when $\operatorname{Lie}(H_\varepsilon)$ preserves a nontrivial coordinate subspace; otherwise, the set is universal.

Notably, this verification can be carried out by examining only the generators in $\mathcal X$ rather than the entire $\operatorname{Lie}\langle\mathcal X\rangle$. The reason is that $\mathcal X$, the Lie algebra $\operatorname{Lie}\langle\mathcal X\rangle$, and $\operatorname{Lie}(H_\varepsilon)$ share the same invariant coordinate subspaces; we defer the technical proof to Lemma~\ref{SM:lem:inheritance} in the Supplemental Material.

We thus arrive at the main structural criterion underlying the algorithm.

\begin{corollary}[Universality via general direction]\label{cor:universality}
Assume that $\mathcal X\subset\mathfrak g$, with $\mathfrak g=\mathfrak u(d)$ or $\mathfrak{su}(d)$, contains a diagonal general direction $X_1$. Then, for a fixed $\varepsilon>0$ satisfying Proposition~\ref{prop:smalleps}, the associated small-step subgroup
\[
H_\varepsilon=\overline{\langle \exp(\varepsilon X_1),\dots,\exp(\varepsilon X_m)\rangle}
\]
is universal, i.e. $H_\varepsilon=G$, if and only if $\mathcal X$ acts irreducibly on the defining representation of $\mathbb C^d$. Equivalently, in the standard basis of $X_1$, universality holds if and only if the generators in $\mathcal X$ have no nontrivial common invariant coordinate subspace.
\end{corollary}
The proof is presented below Corollary~\ref{SM:cor:universality} in Supplemental Material.




\section{The Algorithms}
\label{sec:algorithm}

\noindent This section presents our polynomial-time algorithms. Working in the standard basis of the general element, we test whether the remaining generators share any nontrivial proper coordinate subspace. We show this condition is equivalent to checking connectivity in a certain graph, yielding a simple and efficient criterion for universality. This framework is inherently constructive: when a set is found to be non-universal, the graph's disconnected components directly indicate which skew-Hermitian generators are missing to repair it. Finally, we apply this graph-theoretic framework to prove that universal controllability can be achieved with merely two generators.

\paragraph{Input model.}
Let $G=U(d)$ or $SU(d)$ with corresponding Lie algebra $\mathfrak{g} = \mathfrak{u}(d)$ or $\mathfrak{su}(d)$. Let $\mathcal{X} = \{X_1, \dots, X_m\} \subset \mathfrak{g}$ be a finite set of skew-Hermitian generators. We specifically require the first generator, $X_1 = i\operatorname{diag}(\theta_1, \dots, \theta_d)$, to be a diagonal general direction in $\mathfrak{g}$. The smallest real Lie subalgebra generated by $\mathcal{X}$ is denoted by $$\mathfrak{l} := \operatorname{Lie}_{\mathbb{R}} \langle X_1, \dots, X_m \rangle \subseteq \mathfrak{g}.$$ Fixing a small step size $\varepsilon > 0$ satisfying Proposition~\ref{prop:smalleps}, we define the associated sampled gate set $$S_\varepsilon := \{\exp(\varepsilon X_1), \dots, \exp(\varepsilon X_m)\} \subset G.$$
We denote the closed subgroup generated by this discrete set as $H_\varepsilon := \overline{\langle S_\varepsilon \rangle} \subseteq G$.

\paragraph{Reduction to the standard basis.}
By Lemma~\ref{SM:lem:inheritance}, it suffices to test the generators $\mathcal X$ themselves for invariant subspaces, rather than the full Lie algebra $\mathfrak l$. A priori, however, invariant subspace detection depends on the choice of basis. The presence of the diagonal general direction $X_1$ removes this ambiguity: any invariant subspace must be a coordinate subspace in the standard basis of $X_1$, which here is the standard basis.

\begin{proposition} \label{pro:irreducible-basis}
Let $X_1 = i\operatorname{diag}(\theta_1, \dots, \theta_d) \in \mathfrak{l}$ be a diagonal matrix with a non-degenerate spectrum (mutually distinct $\theta_j$). If $W \subseteq \mathbb{C}^d$ is an invariant subspace of $\mathfrak{l}$, then $W$ is spanned entirely by a subset of the standard basis vectors $\{e_1, \dots, e_d\}$, and its orthogonal complement $W^\perp$ is spanned by the complementary subset.
\end{proposition}
The proof is presented below Proposition~\ref{SM:pro:irreducible-basis} in Supplemental Material.

Thus the universality check reduces to a connectivity problem on the standard basis: one must determine whether the off-diagonal actions of the generators connect all basis vectors into a single component.

\paragraph{Algorithmic detection of universality.}

Accordingly, invariant subspace detection becomes a graph-connectivity problem.  Initializing with a single basis vector, we iteratively apply the generators in $\mathcal{X}$ to enlarge the connectable coordinate subspace. If the generated space ceases to expand before encompassing all of $\mathbb{C}^d$, the system is reducible (non-universal).  If all basis vectors lie in a single connected component, the system is strictly irreducible (universal).

\begin{algorithm}[H]
\caption{Testing universality via invariant subspaces}
\label{AL1}
\begin{algorithmic}[1]
\Require $\mathcal{X}=\{X_1,\dots,X_m\}$ (skew-Hermitian, $X_1$ diagonal), dimension $d$
\Ensure Universality test
\State Choose arbitrary $k \in \{1,\dots,d\}$
\State $I_{curr} \rightarrow \{k\}$
\Repeat
    \State $I_{curr} \rightarrow I_{curr}$
    \For{$\ell \in I_{curr}$}
        \For{$j=2$ to $m$}
            \For{$r=1$ to $d$}
                \If{$\langle e_r | X_j | e_\ell \rangle \neq 0$}
                    \State $I_{curr} \rightarrow I_{curr} \cup \{r\}$
                \EndIf {$\langle e_r | X_j | e_\ell \rangle = 0$}
            \EndFor
        \EndFor
    \EndFor
\Until{$I_{curr} = I_{curr}$}
\If{$|I_{curr}| = d$}
    \State \Return ``$\mathcal{X}$ is Universal ($H_\varepsilon = G$)''
\Else
    \State \Return ``$\mathcal{X}$ is Not Universal (Proper invariant subspace found)''
\EndIf
\end{algorithmic}
\end{algorithm}
Because the generated subspace is spanned by the standard basis, the choice of the initial index $k$ is arbitrary; any choice will successfully identify if the space is disjoint.

\paragraph{Fixing non-universality.}

If Algorithm~\ref{AL1} terminates with a proper subset $I_{curr} \subsetneq \{1, \dots, d\}$, the system lacks the necessary generators to be universal. We can repair this by introducing an elementary skew-Hermitian matrix that bridges the invariant subspace and its orthogonal complement.

\begin{algorithm}[H]
\caption{Repairing non-universality}
\label{AL2}
\begin{algorithmic}[1]
\State \textbf{Input:} The incomplete index set $I_{curr}$ from Algorithm~\ref{AL1}.
\While{$|I_{curr}| < d$}
    \State Arbitrarily select an index $a \in I_{curr}$ and an index $b \in \{1, \dots, d\} \setminus I_{curr}$.
    \State Define the elementary bridging matrix $E_{ab} = |e_a\rangle\langle e_b|$.
    \State Append the skew-Hermitian matrix $Y_{ab} = E_{ab} - E_{ba}$ to the generator set $\mathcal{X}$.
    \State Update the active index set: $I_{curr} \rightarrow I_{curr} \cup \{b\}$.
\EndWhile
\State \Return A universal generator set $\mathcal{X}_{new}$.
\end{algorithmic}
\end{algorithm}
Alternatively, one may use the symmetrization $Y_{ab} = i(E_{ab} + E_{ba}) \in \mathfrak{u}(d)$ depending on the specific physical constraints of the quantum control system.

\medskip

\paragraph{Minimal Generator Sets for Hamiltonian Control.}

The Algorithm~\ref{AL2} implies that two skew-Hermitian generators already suffice to produce a universal gate set.
\begin{proposition}
Let $X_1 \in \mathfrak{g}$ be a diagonal general direction, and choose the control generator to be
\begin{equation}
    X_2 = \sum_{j=1}^{d-1} c_j (E_{j, j+1} - E_{j+1, j}),
\end{equation}
where $E_{j, k}$ denotes the standard matrix unit and the coefficients $c_j \in \mathbb{R}$ are all strictly nonzero. Then $\operatorname{Lie}_{\mathbb{R}} \langle X_1, X_2 \rangle = \mathfrak{g}$.
Then for any $\varepsilon > 0$ satisfying Proposition~\ref{prop:smalleps}, the two-element set $S_\varepsilon = \{\exp(\varepsilon X_1), \exp(\varepsilon X_2)\}$ is universal.
\end{proposition}

A graph-theoretic interpretation of this construction is given in the proof of Proposition~\ref{SM:prop:minimal-generators} in Supplemental Material. In addition, substituting the symmetric off-diagonal generator $X_2^\prime = i \sum_{j=1}^{d-1} c_j (E_{j, j+1} + E_{j+1, j})$ yields the same connected coupling graph. 

Hence we have answered the question on the minimum number of independent generators required to achieve universality: {\textit{two generators are sufficient to yield a universal gate set}}. This conclusion is in line with established quantum-control results relating complete controllability to a drift Hamiltonian together with a single control field \cite{2001Complete,2001Quantum}.

\bigskip

To illustrate these algorithms, let us present a concrete example.

\begin{example}\label{ex:u4-physics}
Consider a two-qubit system with computational basis $\{|00\rangle,|01\rangle,|10\rangle,|11\rangle\},$ and let the Hermitian drift be $H_x^{(1)}=X\otimes I, H_x^{(2)}=I\otimes X$ with control Hamiltonians
\[
H_{\mathrm{drift}}
=\omega_1\, Z\otimes I+\omega_2\, I\otimes Z+J\, Z\otimes Z.
\]
Set $X_{\mathrm{drift}}=-iH_{\mathrm{drift}},
X_x^{(1)}=-iH_x^{(1)},
X_x^{(2)}=-iH_x^{(2)}.$
Assume that $J\neq 0$ and that $X_{\mathrm{drift}}$ satisfies the incommensurability hypothesis of the theorem.

With the reduced control set $\mathcal X=\{X_{\mathrm{drift}},X_x^{(1)}\},$ the coupling graph in the standard basis of $X_{\mathrm{drift}}$ has two connected components, $\{|00\rangle,|10\rangle\}$ and $\{|01\rangle,|11\rangle\},$ so the generated Lie algebra is reducible and the system is not fully controllable.

After adjoining the second local drive $X_x^{(2)}$, the coupling graph becomes connected. Our criterion therefore certifies that $\operatorname{Lie}_{\mathbb R}
\langle X_{\mathrm{drift}},X_x^{(1)},X_x^{(2)}\rangle
=\mathfrak{su}(4),$ that is, the system is fully controllable up to global phase. This example shows how, for a generic drift-plus-drive architecture, universality can be diagnosed by a simple connectivity condition on the coupling graph.
\end{example}

\section{Conclusion}

In this Letter, we present a concise Lie-algebraic criterion for single‑qudit universality: the existence of a diagonal Hamiltonian with rationally independent eigenvalues and the ability to exponentiate these generators with sufficiently small parameters. This yields two results: (1) a polynomial‑time test for universality — simply check whether the defining representation admits a nontrivial invariant subspace; (2) a systematic method to repair a non‑universal set, together with the proof that two generators suffice for universal control. Moreover, our analysis directly links to hardware‑native Hamiltonians, thereby providing a theoretical foundation for the practical construction of quantum circuits.
\smallskip

\paragraph*{ Acknowledgments.}
We thank Fang Kun for insightful suggestions which help to improve the exposition of the paper.
This work is supported by grants NSFC 12341101 and G02X2403006.

\bibliography{ref}

\clearpage

\setcounter{equation}{0}
\setcounter{figure}{0}
\setcounter{table}{0}
\setcounter{page}{1}
\setcounter{section}{0}
\renewcommand{\theequation}{S\arabic{equation}}
\renewcommand{\thefigure}{S\arabic{figure}}
\renewcommand{\thesection}{S\Roman{section}}
\renewcommand{\thepage}{S\arabic{page}}

\newtheorem{supptheorem}{Theorem}[section]
\renewcommand{\thesupptheorem}{S\arabic{supptheorem}}
\newtheorem{supplemma}[supptheorem]{Lemma}
\newtheorem{suppprop}[supptheorem]{Proposition}
\newtheorem{suppcorollary}[supptheorem]{Corollary}
\newtheorem{suppexample}[supptheorem]{Example}

\begin{center}
\vspace*{.5\baselineskip}
{\textbf{\large Supplemental Material}}\\[1pt] \quad \\
\end{center}

\vspace{1cm}
In this Supplemental Material, we provide the detailed proofs and extended derivations supporting the results in the main text. We may reiterate some of the steps to ensure that the Supplemental Material are explicit and self-contained.

\hypersetup{linkcolor=blue} 
\tableofcontents

\vspace{1cm}

\section{Preliminaries and notations}

Throughout the Supplemental Material, let $G$ ve a compact Lie group, and $\mathfrak{g}=\operatorname{Lie}(G)$ be the corresponding Lie algebra.
Let $\mathcal X=\{X_1,\ldots,X_m\}\subset \mathfrak g$ be a finite set of skew-Hermitian generators. We write
\[
\mathfrak l
:=
\operatorname{Lie}_{\mathbb R}\langle \mathcal X\rangle
=
\operatorname{Lie}_{\mathbb R}\langle X_1,\ldots,X_m\rangle
\subseteq \mathfrak g
\]
for the real Lie algebra generated by \(\mathcal X\).

For a sampling parameter \(\varepsilon>0\), define the corresponding finite gate set
\[
S_\varepsilon
:=
\{\exp(\varepsilon X_1),\ldots,\exp(\varepsilon X_m)\}
\subset G,
\]
$\langle S_\varepsilon\rangle$ be the set generated by group multiplication based on these sampled gates of $S_\varepsilon$. Denote
\[
H_\varepsilon
:=
\overline{\langle S_\varepsilon\rangle}
\subseteq G
\]
the closure of $\langle S_\varepsilon\rangle$.

We specify $G=U(d)$ or $G=SU(d)$ and corresponding $\mathfrak{g}=\mathfrak{u}(d)$ or $\mathfrak{g}=\mathfrak{su}(d)$.

As mentioned below Definition~\ref{def:GeneralElements}, a diagonal generator \(X_1\) with incommensurate spectrum plays a central role since its sampled exponential generates a dense subgroup of the standard maximal torus, according to the Kronecker-Weyl density theorem \cite{BroeckerTomDieck1985}. We use \(\{e_1,\ldots,e_d\}\) for the standard basis of \(\mathbb C^d\), and \(E_{jk}\) for the standard matrix unit.

We then move to the notion of generality of the elements of a Lie group, which is already defined in Definition~\ref{def:GeneralElements}. An element $g \in G$ is said to be general if $\langle g \rangle=T_{G}$, namely, $g$ generates a maximal torus of $G$. As mentioned in the main text, we are interested in the situation whose parameter \(\varepsilon>0\) makes \(\exp(\varepsilon X_1)\) general. The $X_1$ is called general direction. For a diagonal \(X_1=i\operatorname{diag}(\theta_1,\ldots,\theta_d)\), this implies that the angles $1, \theta_1,\ldots,\theta_d$ are linear independent over $\mathbb{Q}$.

\section{Small-step generation and connectedness}

This section explains why the universality of a sampled gate set \(S_\epsilon\) can be faithfully tested through the Lie-algebraic data \(\mathcal X\) provided that the parameter $\varepsilon$ is sufficiently small. The key point is to avoid the disconnectness.

First, we show that the Lie algebra $\mathfrak{l}$ generated by \(\mathcal X\) is contained in the Lie algebra of the closed subgroup \(H_\varepsilon\) generated by the sampled gates. If the infinitesimal generators already generate the full target algebra \(\mathfrak g\), then the sampled gates are universal.

We then prove a complementary connectedness statement: in the small-step regime, the generators cannot jump into disconnected permutation components of a maximal-rank subgroup. This is the technical point that allows the maximal-rank classification to be applied at the Lie-algebra level.

\subsection{Small-step Lie-algebra generation}

\begin{suppprop}[Small-step Lie-algebra generation]\label{SM:prop:small-step}
Let $G$ be either $U(d)$ or $SU(d)$ with Lie algebra $\mathfrak{g}$, and let $\mathcal{X} = \{X_1, \dots, X_m\} \subset \mathfrak{g}$ be a finite set of skew-Hermitian matrices. Let $\mathfrak{l}$ be the smallest real Lie subalgebra of $\mathfrak{g}$ generated by $\mathcal{X}$ through Eq.\eqref{eq:GeneratingLieAlgebra}. For a given $\varepsilon > 0$, define the set of discrete generators
\[
S_\varepsilon := \{ \exp(\varepsilon X_1), \dots, \exp(\varepsilon X_m) \} \subset G,
\]
and let $H_\varepsilon := \overline{\langle S_\varepsilon \rangle}$ be the closed subgroup generated by $S_\varepsilon$. For sufficiently small $\varepsilon > 0$, $H_\varepsilon$ is a closed Lie subgroup of $G$ whose Lie algebra satisfies $\mathfrak{l} \subseteq \text{Lie}(H_\varepsilon)$. Consequently, if the generators span the full algebra ($\mathfrak{l} = \mathfrak{g}$), then the gate set is universal, i.e., $H_\varepsilon = G$.
\end{suppprop}
\begin{proof}
By construction, $H_\varepsilon$ is a closed subgroup of $G$. By Cartan's closed subgroup theorem \cite{lee2003smooth}, any closed subgroup of a Lie group is an embedded Lie subgroup.

Because the exponential map is a local diffeomorphism near the origin, we can choose $\varepsilon > 0$ sufficiently small such that the elements $\exp(\varepsilon X_j)$ reside within a neighborhood of the identity where the principal logarithm is well-defined. For sufficiently small $\varepsilon>0$, each $\exp(\varepsilon X_j)$ lies in a neighborhood of the identity where the logarithm is well defined. Hence $\varepsilon X_j\in\operatorname{Lie}(H_\varepsilon)$, and therefore $X_j\in\operatorname{Lie}(H_\varepsilon)$ for all $j$.

Since $\text{Lie}(H_\varepsilon)$ is a real Lie subalgebra of $\mathfrak{g}$ containing the initial generators $\mathcal{X}$, it must contain the smallest real Lie subalgebra generated by them. Therefore, $\mathfrak{l} \subseteq \text{Lie}(H_\varepsilon)$. 

Finally, if $\mathfrak{l} = \mathfrak{g}$, then $\text{Lie}(H_\varepsilon) = \mathfrak{g}$. Because the groups $U(d)$ and $SU(d)$ are connected, any Lie subgroup possessing the full Lie algebra $\mathfrak{g}$ must be the entire group, yielding $H_\varepsilon = G$.
\end{proof}

\subsection{Proof of Proposition~\ref{prop:smalleps} (Small-step connectedness)}

$U(d)$ contains subgroups that are finite, for example, the permutation subgroup and its subgroups. Since the elements of such subgroups are members of $U(d)$ as well, they can be probably obtained by exponentiation under certain $X_i \in \mathcal{X}$ with particular value of the parameter $\varepsilon$. In that case, the sample group would be finite thus non-universal. To avoid the possibility, we require the $\varepsilon$ sufficient small.

\begin{suppprop}[Small-step exclusion of disconnected components]
\label{SM:prop:smalleps}
Let $H \subset U(d)$ or $SU(d)$ be a closed maximal-rank subgroup with identity component $H^0$. If $H \neq H^0$, then every element $g \in H \setminus H^0$ satisfies $\|g - I\| \ge \sqrt{2}$ in the operator norm. Consequently, if $\varepsilon > 0$ is sufficiently small such that all generators satisfy $\|\exp(\varepsilon X_i) - I\| < \sqrt{2}$, then the generated subgroup $H_\varepsilon = \overline{\langle S_\varepsilon \rangle}$ is connected, meaning $H_\varepsilon = H_\varepsilon^0$.
\end{suppprop}

\begin{proof}
If $H \neq H^0$, then $H$ is obtained from the connected block-diagonal subgroup $H^0$ by adjoining a finite group of block permutations \cite{BorevichKrupetskii1981, AntoneliForgerGaviria2012}. Thus, any element $g \in H \setminus H^0$ acts as a nontrivial permutation $p$ on the orthogonal block decomposition $\mathbb{C}^d = \bigoplus_k V_k$ defining $H^0$. 

Because the permutation is nontrivial, there exists a unit vector $v \in V_a$ such that $gv \in V_b$ with $b \neq a$. Since $g$ is a unitary operator, it preserves the vector norm, meaning $\|gv\| = \|v\| = 1$. Furthermore, because the invariant blocks are mutually orthogonal ($V_a \perp V_b$), the Pythagorean theorem yields
\begin{equation}
    \|gv - v\| = \sqrt{\|gv\|^2 + \|v\|^2} = \sqrt{1 + 1} = \sqrt{2}.
\end{equation}
By the definition of the operator norm, $\|g - I\| \ge \|(g - I)v\| = \|gv - v\| = \sqrt{2}$.

For the continuous generators, expanding the exponential yields $\exp(\varepsilon X_i) = I + \varepsilon X_i + O(\varepsilon^2)$. Because the set of generators $\mathcal{X}$ is finite, we can choose $\varepsilon > 0$ sufficiently small such that the condition $\|\exp(\varepsilon X_i) - I\| < \sqrt{2}$ is satisfied for all $i$ simultaneously. Under this small-step constraint, none of the generators can lie in the disconnected components $H \setminus H^0$, meaning they must all reside within $H^0$. Since $H^0$ is a closed subgroup, taking the closure of the generated set gives $H_\varepsilon \subseteq H^0$. Therefore, $H_\varepsilon$ is connected.
\end{proof}

\medskip

\section{Maximal-rank structure for unitary groups}

\subsection{The Borel–de Siebenthal for unitary groups}
We first summarize the classification of maximal-rank subgroups in the unitary case; see Borevich-Krupetskii~\cite{BorevichKrupetskii1981} and Borel--de Siebenthal~\cite{BorelDeSiebenthal1949} for details.

\begin{suppprop}[Maximal-rank subgroups of $U(d)$]
Up to conjugacy, every connected closed subgroup of maximal rank in
$U(d)$ that contains $D_d$ is a block subgroup of the form
\[
U(d_1)\times \cdots \times U(d_k) \subset U(d),
\quad
d_1+\cdots+d_k=d,\quad d_\ell \geq 1.
\]
Then:
\begin{enumerate}
    \item[(a)] The identity component of any maximal-rank subgroup
    $H\subset U(d)$ containing $D_d$ is, up to conjugacy, of the form 
    $$H^0 \simeq U(d_1,\ldots,d_k).$$

    \item[(b)] The full subgroup $H$ is the normalizer of $H^0$ in
    $U(d)$:
    \[
    H=N_{U(d)}(H^0)=H^0\rtimes F,
    \]
    where $F$ is a finite group generated by permutations of those
    blocks $U(d_\ell)$ which have the same dimension. In particular,
    $H$ is connected if and only if all block sizes $d_\ell$ are
    distinct. Moreover, the Lie algebra of $H$ is always
    \[
    \operatorname{Lie}(H)=\operatorname{Lie
    }(H^0)\simeq
    \mathfrak{u}(d_1)\oplus\cdots\oplus\mathfrak{u}(d_k).
    \]
\end{enumerate}

Thus, there are two basic types of maximal subgroups of maximal rank
in $U(d)$: connected block subgroups $U(p)\times U(q)$ with $p\neq q$,
and nonconnected block-monomial normalizers when equal-size blocks occur.

For $SU(d)$, type $A_{d-1}$, the situation is analogous, with the
additional determinant-one constraint.
\end{suppprop}

\begin{suppprop}[Maximal-rank subgroups of $SU(d)$]
Let
\[
T_{SU(d)}=D_d\cap SU(d)
\]
be the standard maximal torus in $SU(d)$. Up to conjugacy, every closed
subgroup of maximal rank in $SU(d)$ that contains $T_{SU(d)}$ has
identity component of the form
\[
H^0\simeq S\bigl(U(d_1)\times\cdots\times U(d_k)\bigr),
\quad
d_1+\cdots+d_k=d,\quad d_\ell\geq 1,
\]
where $S\bigl(U(d_1)\times\cdots\times U(d_k)\bigr):=
\operatorname{diag}(A_1,\ldots,A_k)$,
with
\[
A_\ell\in U(d_\ell),
\qquad
\prod_{\ell=1}^k \det(A_\ell)=1.
\]
As in the unitary case, the full subgroup is the normalizer
\[
H=N_{SU(d)}(H^0)=H^0\rtimes F,
\]
where $F$ is a finite group of block permutations preserving the
determinant-one condition. In particular, $H$ is connected if and only
if all block sizes $d_\ell$ are distinct. The Lie algebra of $H$ is
\[
\operatorname{Lie}(H)=\operatorname{Lie}(H^0)\simeq
\mathfrak{su}(d_1)\oplus\cdots\oplus\mathfrak{su}(d_k)
\oplus \mathrm{i}\mathbb{R}^{k-1}.
\]
\end{suppprop}

\subsection{Proof of Theorem~\ref{thm:block} (Block-diagonal structure)}

We next restrict to subgroups containing the standard maximal torus, and give a coordinate-level version of this statement. 
\begin{supptheorem}[Block-diagonal structure]\label{SM:thm:block}
Let $H \subset G$ be a connected closed subgroup of $G = U(d)$ [resp. $SU(d)$]. If $H$ contains the standard maximal torus $T_G$, then there exists a partition $d = d_1 + \dots + d_r$ and a permutation matrix $P$ (representing a reordering of the standard basis) such that the conjugated subgroup $\tilde{H} = P H P^{-1}$ is exactly the block-diagonal subgroup:
\begin{equation}
    \tilde{H} = \prod_{j=1}^r U(d_j) \quad \left[\text{resp. } S\left(\prod_{j=1}^r U(d_j)\right)\right].
\end{equation}
Furthermore, the Lie algebra of $\tilde{H}$, given by $\text{Lie}(\tilde{H}) = P \text{Lie}(H) P^{-1}$, exhibits the identical block-diagonal structure:
\begin{equation}
    \text{Lie}(\tilde{H}) = \bigoplus_{j=1}^{r}\mathfrak{u}(d_j) \quad \left[\text{resp. } \bigoplus_{j=1}^{r}\mathfrak{su}(d_j) \oplus i\mathbb{R}^{r-1}\right]. 
\end{equation}
In particular, any proper such subgroup preserves a nontrivial coordinate subspace.
\end{supptheorem}

\begin{proof}
Since $H$ is a connected closed subgroup of $G$, it is uniquely and completely determined by its real Lie algebra $\mathfrak{h} = \text{Lie}(H) \subseteq \mathfrak{g} = \text{Lie}(G)$. By hypothesis, $H$ contains the standard maximal torus $T_G$, which implies that $\mathfrak{h}$ contains the standard Cartan subalgebra $\mathfrak{t} \subset \mathfrak{g}$, consisting of all diagonal skew-Hermitian matrices. Because $\mathfrak{h}$ contains the maximal abelian subalgebra $\mathfrak{t}$, $H$ is a maximal-rank subgroup of $G$.

To analyze the structure of $\mathfrak{h}$, it is highly advantageous to pass to the complexified Lie algebras. Let $\mathfrak{g}_{\mathbb{C}} = \mathfrak{g} \otimes \mathbb{C}$ and $\mathfrak{h}_{\mathbb{C}} = \mathfrak{h} \otimes \mathbb{C}$. For $G = U(d)$ [resp. $SU(d)$], the complexification is $\mathfrak{g}_{\mathbb{C}} = \mathfrak{gl}_d(\mathbb{C})$ [resp. $\mathfrak{sl}_d(\mathbb{C})$], and the complexified Cartan subalgebra $\mathfrak{t}_{\mathbb{C}}$ consists of all diagonal matrices [resp. traceless diagonal matrices].

The standard root space decomposition of $\mathfrak{g}_{\mathbb{C}}$ with respect to $\mathfrak{t}_{\mathbb{C}}$ is given by:
\begin{equation}
    \mathfrak{g}_{\mathbb{C}} = \mathfrak{t}_{\mathbb{C}} \oplus \bigoplus_{j \neq k} \mathfrak{g}_{\alpha_{j,k}},
\end{equation}
where the roots are $\alpha_{j,k} = e_j - e_k$ (with $e_j$ being the linear functional extracting the $j$-th diagonal entry), and the corresponding one-dimensional root spaces are spanned by the standard matrix units, $\mathfrak{g}_{\alpha_{j,k}} = \mathbb{C}E_{j,k}$.

Because $\mathfrak{h}_{\mathbb{C}}$ is a Lie subalgebra that fully contains $\mathfrak{t}_{\mathbb{C}}$, its root space decomposition must be a restriction of the decomposition of $\mathfrak{g}_{\mathbb{C}}$. Therefore, $\mathfrak{h}_{\mathbb{C}}$ must take the form:
\begin{equation}
    \mathfrak{h}_{\mathbb{C}} = \mathfrak{t}_{\mathbb{C}} \oplus \bigoplus_{(j,k) \in \Phi_H} \mathbb{C}E_{j,k},
\end{equation}
where $\Phi_H$ is a subset of the roots of $\mathfrak{g}_{\mathbb{C}}$. Since $\mathfrak{h}$ is the Lie algebra of a compact group, $\Phi_H$ must be symmetric: if $(j,k) \in \Phi_H$, then $(k,j) \in \Phi_H$ because the generators of the compact real form require combinations of $E_{j,k} - E_{k,j}$. Furthermore, because $\mathfrak{h}_{\mathbb{C}}$ is closed under the Lie bracket, if $(j,k) \in \Phi_H$ and $(k,l) \in \Phi_H$ (with $j \neq l$), the commutator $[E_{j,k}, E_{k,l}] = E_{j,l}$ implies that $(j,l) \in \Phi_H$.

We can now define a binary relation $\sim$ on the basis indices $\{1, 2, \dots, d\}$ as follows: $j \sim k$ if and only if $j = k$ or $(j,k) \in \Phi_H$. 
The properties of $\Phi_H$ established above guarantee that $\sim$ is an equivalence relation:
\begin{enumerate}
    \item \textbf{Reflexivity:} $j \sim j$ by definition.
    \item \textbf{Symmetry:} $j \sim k \implies k \sim j$ (since $\Phi_H$ is symmetric).
    \item \textbf{Transitivity:} $j \sim k$ and $k \sim l \implies j \sim l$ (since $\mathfrak{h}_{\mathbb{C}}$ is closed under commutators).
\end{enumerate}

This equivalence relation inherently partitions the set of indices $\{1, \dots, d\}$ into $r$ disjoint equivalence classes, with respective sizes $d_1, d_2, \dots, d_r$ such that $\sum_{m=1}^r d_m = d$. 

Let $P$ be a permutation matrix that reorders the standard basis of $\mathbb{C}^d$ such that indices belonging to the same equivalence class are grouped contiguously. Applying this similarity transformation to the algebra, $\tilde{\mathfrak{h}}_{\mathbb{C}} = P \mathfrak{h}_{\mathbb{C}} P^{-1}$, forces all non-zero root spaces to reside within dense $d_m \times d_m$ blocks along the diagonal. Since all possible root vectors within each equivalence class are present, each block constitutes a full general linear algebra $\mathfrak{gl}_{d_m}(\mathbb{C})$.

Returning to the compact real form by intersecting with $\mathfrak{u}(d)$ [resp. $\mathfrak{su}(d)$], the conjugated real Lie algebra $\tilde{\mathfrak{h}} = P \mathfrak{h} P^{-1}$ perfectly matches the block-diagonal structure:
\begin{equation}
    \text{Lie}(\tilde{H}) = \bigoplus_{m=1}^{r}\mathfrak{u}(d_m) \quad \left[\text{resp. } \bigoplus_{m=1}^{r}\mathfrak{su}(d_m) \oplus i\mathbb{R}^{r-1}\right].
\end{equation}
(Note: In the $SU(d)$ case, the $i\mathbb{R}^{r-1}$ term accounts for the traceless diagonal elements that commute with all $\mathfrak{su}(d_m)$ blocks, representing the center of the subgroup).

Finally, because $H$ is connected, the Lie group is entirely generated by the exponential map of its Lie algebra, $H = \exp(\mathfrak{h})$. Since the exponential map commutes with conjugation, $\tilde{H} = \exp(\tilde{\mathfrak{h}})$. Exponentiating a direct sum of Lie algebras yields the direct product of their corresponding Lie groups, thus establishing:
\begin{equation}
    \tilde{H} = \prod_{m=1}^r U(d_m) \quad \left[\text{resp. } S\left(\prod_{m=1}^r U(d_m)\right)\right].
\end{equation}
This completes the proof.
\end{proof}

\begin{suppexample}\label{SM:ex:U(3) block-diag}
 In $U(3)$ the subgroup
\[
U(2)\times U(1)
=
\left\{
\begin{pmatrix}
A&0\\
0&z
\end{pmatrix}
:
A\in U(2),\ z\in U(1)
\right\}
\]
contains the standard maximal torus and preserves the subspace $\operatorname{span}\{e_1,e_2\}\subset\mathbb C^3$. 
Likewise, in $SU(3)$ the subgroup $S(U(2)\times U(1))$ is proper, maximal rank, and reducible in the defining representation.   
\end{suppexample}

\medskip

\section{Invariant subspaces and reduction to coordinate subspaces}

This section contains the two lemmas needed for the algorithmic criterion. We first show that invariant coordinate subspaces are preserved when passing between $\mathcal X$, $\mathfrak l$, and $\text{Lie}(H_\varepsilon)$, and then prove that, in the presence of a diagonal generator with simple spectrum, every invariant subspace is a coordinate subspace.

\subsection{Inheritance of reducibility}

\begin{supplemma}[Inheritance of reducibility]\label{SM:lem:inheritance}
Let $\mathcal{X} = \{X_1, \ldots, X_m\} \subset \mathfrak{g}$ be a finite set of skew-Hermitian matrices, let $\mathfrak{l}$ be the real Lie algebra generated by $\mathcal{X}$, and let $H_\varepsilon = \overline{\langle S_\varepsilon \rangle}$. The following statements are equivalent:
\begin{enumerate}
    \item[(i)] The finite set of generators $\mathcal{X}$ shares a non-trivial invariant coordinate subspace.
    \item[(ii)] The Lie algebra $\mathfrak{l}$ shares the same invariant coordinate subspace.
    \item[(iii)] The Lie algebra of the generated group, $\text{Lie}(H_\varepsilon)$, shares the same invariant coordinate subspace.
\end{enumerate}
\end{supplemma}

\begin{proof}
(i) $\iff$ (ii): The forward implication holds because any subspace invariant under matrices $X_i$ and $X_j$ is automatically invariant under their Lie bracket $[X_i, X_j] = X_i X_j - X_j X_i$; by induction, this invariance extends to the entire generated real Lie algebra $\mathfrak{l}$. The converse is trivial since $\mathcal{X} \subset \mathfrak{l}$.

(i) $\implies$ (iii): Suppose $\mathcal{X}$ shares an invariant coordinate subspace, meaning there exists a permutation matrix $P$ such that $P X_j P^{-1}$ is block-diagonal for all $j=1,\dots,m$. By the naturality of the exponential map under conjugation, 
\begin{equation*}
    P \exp(\varepsilon X_j) P^{-1} = \exp(\varepsilon P X_j P^{-1}).
\end{equation*}
This identity ensures that the discrete generators of the group are also block-diagonal under conjugation by $P$. Because matrix multiplication and topological closures preserve block-diagonal structures, the conjugated closed subgroup $\tilde{H}_\varepsilon = P H_\varepsilon P^{-1}$ is strictly block-diagonal. Consequently, its Lie algebra $\text{Lie}(\tilde{H}_\varepsilon) = P \text{Lie}(H_\varepsilon) P^{-1}$ is block-diagonal, meaning $\text{Lie}(H_\varepsilon)$ preserves the exact same invariant subspace.

(iii) $\implies$ (i): Proposition~\ref{prop:smalleps} establishes that for sufficiently small $\varepsilon > 0$, the absence of disconnected components ensures the matrix logarithm maps the generators directly into the algebra, yielding $\varepsilon X_j \in \text{Lie}(H_\varepsilon)$, and thus $X_j \in \text{Lie}(H_\varepsilon)$ for all $j$. If $\text{Lie}(H_\varepsilon)$ preserves an invariant coordinate subspace, then every element contained within it, including the finite subset $\mathcal{X}$, must preserve it as well.
\end{proof}

\medskip

\subsection{Proof of Corollary~\ref{cor:universality}  (Universality criterion)}

\begin{suppcorollary}[Universality via general direction]\label{SM:cor:universality}
Assume $X_1 \in \mathcal{X}$ is a diagonal general direction such that its continuous evolution generates the standard maximal torus, i.e., $\overline{\langle \exp(\varepsilon X_1) \rangle} = T_G$ for a generic $\varepsilon > 0$. 

For sufficiently small $\varepsilon$ satisfying Proposition~\ref{prop:smalleps}, the generated subgroup is universal ($H_\varepsilon = G$) if and only if the set of generators $\mathcal{X}$ acts irreducibly on the standard basis of $\mathbb{C}^d$. Equivalently, universality is achieved if and only if the matrices in $\mathcal{X}$ share no non-trivial invariant coordinate subspace.
\end{suppcorollary}

\begin{proof}
By the general direction assumption, the closure of the cyclic subgroup generated by $\exp(\varepsilon X_1)$ is exactly the standard maximal torus, yielding $T_G \subset H_\varepsilon$. This guarantees that $H_\varepsilon$ is a closed subgroup of maximal rank. Furthermore, the small-step condition in Proposition~\ref{prop:smalleps} ensures that the generated subgroup is connected ($H_\varepsilon = H_\varepsilon^0$). Therefore, $H_\varepsilon$ satisfies all the preconditions of Theorem~\ref{thm:block}. 

According to Theorem~\ref{thm:block}, there is a strict dichotomy: either $H_\varepsilon$ is the full group $G$, or its Lie algebra $\text{Lie}(H_\varepsilon)$ is a proper block-diagonal subalgebra (up to a basis permutation $P$). By Lemma~\ref{SM:lem:inheritance}, $\text{Lie}(H_\varepsilon)$ exhibits this proper block-diagonal structure if and only if the finite set of generators $\mathcal{X}$ shares a non-trivial invariant coordinate subspace. Consequently, if no such common invariant subspace exists, $\text{Lie}(H_\varepsilon)$ cannot be a proper block-diagonal subalgebra. Because $H_\varepsilon$ is a connected maximal-rank subgroup, this exclusion leaves $\text{Lie}(H_\varepsilon) = \mathfrak{g}$ as the only remaining possibility, thereby forcing $H_\varepsilon = G$.
\end{proof}

\medskip
\subsection{Proof of Proposition~\ref{pro:irreducible-basis} (Reduction to coordinate subspaces)}
\begin{suppprop} \label{SM:pro:irreducible-basis}
Let $X_1 = i\operatorname{diag}(\theta_1, \dots, \theta_d) \in \mathfrak{l}$ be a diagonal matrix with a non-degenerate spectrum (mutually distinct $\theta_j$). If $W \subseteq \mathbb{C}^d$ is an invariant subspace of $\mathfrak{l}$, then $W$ is spanned entirely by a subset of the standard basis vectors $\{e_1, \dots, e_d\}$, and its orthogonal complement $W^\perp$ is spanned by the complementary subset.
\end{suppprop}

\begin{proof}
Because $W$ is an invariant subspace of $\mathfrak{l}$, it is invariant under $X_1$. The standard basis vectors $\{e_1, \dots, e_d\}$ are exactly the eigenvectors of $X_1$, corresponding to the distinct eigenvalues $i\theta_j$. It is a standard result of linear algebra that any subspace invariant under a diagonalizable matrix with distinct eigenvalues must be spanned by a subset of its eigenvectors. Therefore, $W = \operatorname{span}\{e_k\}_{k \in I}$ for some index subset $I \subseteq \{1, \dots, d\}$. The orthogonality of the standard basis immediately implies $W^\perp = \operatorname{span}\{e_k\}_{k \notin I}$.
\end{proof}

\medskip

\section{Minimal-generator construction}
This section gives the constructive two-generator universality result. The proof is expressed in graph-theoretic terms and shows that a single chain of nearest-neighbor couplings suffices once a diagonal general direction is present.

\begin{suppprop}[Two-generator universality construction]\label{SM:prop:minimal-generators}
Let $G$ be $U(d)$ or $SU(d)$ with corresponding Lie algebra $\mathfrak{g}$. Let $\mathcal{X} = \{X_1, X_2\} \subset \mathfrak{g}$. Assume $X_1$ is a diagonal general direction. Let the second generator be defined by the sum of all simple root generators:
\begin{equation}
    X_2 = \sum_{j=1}^{n-1} c_j (E_{j, j+1} - E_{j+1, j}),
\end{equation}
where $E_{j, k}$ denotes the standard matrix unit and $c_j \neq 0$ are real coefficients. For any $\varepsilon > 0$ satisfying Proposition~\ref{prop:smalleps}, let $S_\varepsilon = \{\exp(\varepsilon X_1), \exp(\varepsilon X_2)\}$. Then $S_\varepsilon$ is universal, meaning $H_\varepsilon := \overline{\langle S_\varepsilon \rangle} = G$.
\end{suppprop}

\begin{proof}
Because $X_1$ is a diagonal general direction, Corollary~\ref{cor:universality} establishes that the generated subgroup $H_\varepsilon$ is equal to $G$ if and only if the remaining generator $X_2$ acts irreducibly on the standard basis of $\mathbb{C}^n$. The matrix $X_2$ is strictly tridiagonal, possessing non-zero elements exclusively on its first superdiagonal and first subdiagonal. Evaluating its action on the standard basis yields $X_2 e_k \in \operatorname{span}\{e_{k-1}, e_{k+1}\}$ for all $k$ (where we conventionally define $e_0 = e_{n+1} = 0$). 

In the context of Algorithm~\ref{AL1}, this action defines a linear graph where every adjacent basis vector is transitively linked: $e_1 \leftrightarrow e_2 \leftrightarrow \dots \leftrightarrow e_n$. Because this graph forms a single connected component encompassing all $n$ vertices, no proper subset of the standard basis is closed under the action of $X_2$. This means no non-trivial invariant coordinate subspace exists. By the equivalence established in Corollary~\ref{cor:universality}, the system is strictly irreducible, which forces $H_\varepsilon = G$.
\end{proof}

\begin{remark}
The fact that two generators can generate the entire Lie algebra, holds for any complex simple Lie groups \cite{huang2024bases}, which can be proved by using nonvanishing Vandermonde determinants.
\end{remark}

Since a single generator produces only an abelian one-parameter subgroup, two generators are minimal for universality.




\section{Additional Examples}

We conclude with a simple $U(3)$ example illustrating the repair procedure.
\begin{example}\label{ex:u3-repair}
Let $G=U(3)$ with $\mathfrak{g}=\mathfrak{u}(3)$ and standard basis $\{e_1, e_2, e_3\}$. Let $\mathcal{X} = \{X_1, X_2\}$, where $X_1 = i\operatorname{diag}(\sqrt{2}, \sqrt{3}, \sqrt{5})$ serves as the diagonal general direction, and the control Hamiltonian is
\[
X_2 = \begin{pmatrix} 0 & 1 & 0 \\ -1 & 0 & 0 \\ 0 & 0 & 0 \end{pmatrix} \in \mathfrak{su}(3).
\]
The drift Hamiltonian induces relative phase accumulations, while the control Hamiltonian yields a generalized rotation in the $\{e_1, e_2\}$ subspace:
\[
\exp(\epsilon X_2) = \begin{pmatrix} \cos\epsilon & \sin\epsilon & 0 \\ -\sin\epsilon & \cos\epsilon & 0 \\ 0 & 0 & 1 \end{pmatrix}.
\]
Applying Algorithm~\ref{AL1} to $\mathcal{X}$, we initialize $I_{curr} = \{1\}$. Evaluating the matrix elements of $X_2$, we find $X_2 e_1 = -e_2$ and $X_2 e_2 = e_1$. The index set updates to $I_{curr} = \{1, 2\}$. On the next iteration, applying $X_2$ to $\{e_1, e_2\}$ yields no new basis vectors, and the algorithm terminates. Because $I_{curr} \subsetneq \{1, 2, 3\}$, the system is reducible, preserving $W = \operatorname{span}\{e_1, e_2\}$, and is thus not universal.

To fix this via Algorithm~\ref{AL2}, we select $a=2 \in I_{curr}$ and $b=3 \notin I_{curr}$, and append the bridging matrix:
\[
Y_{23} = E_{23} - E_{32} = \begin{pmatrix} 0 & 0 & 0 \\ 0 & 0 & 1 \\ 0 & -1 & 0 \end{pmatrix} \in \mathfrak{su}(3).
\]
Re-running the universality check with $\mathcal{X}^\prime = \{X_1, X_2, Y_{23}\}$, the index set expands transitively: $\{1\} \xrightarrow{X_2} \{1, 2\} \xrightarrow{Y_{23}} \{1, 2, 3\}$. The algorithm terminates with $\vert I_{curr}\vert = 3$, confirming that $\operatorname{Lie}_{\mathbb{R}}\langle X_1, X_2, Y_{23} \rangle$ is $\mathfrak{u}(3)$, and the augmented set is universal.
\end{example}

\end{document}